\documentclass[12pt]{article}

\usepackage{amsmath}
\usepackage{amssymb}
\usepackage{fullpage}
\usepackage{authblk}
\usepackage{enumerate}
\usepackage{paralist}
\usepackage{color,epsfig}
\usepackage{subfigure}

\usepackage{todonotes}

\newtheorem{theorem}{Theorem}
\newtheorem{lemma}[theorem]{Lemma}
\newtheorem{proposition}[theorem]{Proposition}

\newcommand{\qed}{\nobreak $\square$}
\newenvironment{proof}[1][Proof]{\begin{trivlist}\item[\hskip \labelsep {\bfseries #1}]}{\qed\end{trivlist}}

\newcommand{\G}{\mathcal{G}}
\newcommand{\E}{\mathcal{E}}

\newcommand{\ZZ}{\mathbb{Z}}

\newcommand{\FF}{\mathbb{F}}
\newcommand{\TT}{\mathbb{T}}

\newcommand{\calP}{\mathcal{P}}
\newcommand{\calS}{\mathcal{S}}

\newcommand{\QQ}{\mathbb{Q}}

\newcommand{\match}{P}

\newcommand{\cee}{(0,0)}
\newcommand{\ceo}{(0,1)}
\newcommand{\coe}{(1,0)}
\newcommand{\coo}{(1,1)}
\newcommand{\mee}{(0,0)}

\newcommand{\moe}{(1,0)}
\newcommand{\moo}{(1,1)}

\newlength{\mylength}
\newenvironment{myquote}{\list{}{\leftmargin=\mylength\rightmargin=0.0in}\item[]}{\endlist}

\usepackage{ifpdf}

\begin{document}

\title{Types of perfect matchings in toroidal square grids}

\author{%
Marcos Kiwi\thanks{Depto.~Ing.~Matem\'{a}tica \&
  Ctr.~Modelamiento Matem\'atico (CNRS UMI 2807), U.~Chile. 
  Beauchef 851, Santiago, Chile. 
  \texttt{mkiwi@dim.uchile.cl}. 
  Gratefully acknowledges the support of 
  Millennium Nucleus Information and Coordination in Networks ICM/FIC RC130003
  and CONICYT via Basal in Applied Mathematics.}
\and
Martin Loebl\thanks{Dept.~of Applied Mathematics (KAM),
  Charles U.,
  Malostransk\'{e} n\'{a}m. 25, 118~00~Praha~1,
  Czech Republic.
  Gratefully acknowledges the support of the Czech Science Foundation under the contract number P202-13-21988S.
  \texttt{loebl@kam.mff.cuni.cz}
}}

\maketitle              

\begin{abstract}
Let $T_{m,n}$ be toroidal square grid of size $m\times n$ and let both $m$ and
  $n$ be even. Let $P$ be a perfect matching of $T_{m,n}$ and let $D(P)$ be the cycle-rooted spanning forest of $P$ obtained by the generalized Temperley's construction. The types of $P$ and $D(P)$ in
the first homology group $H_1(\TT,\ZZ)$ of torus $\TT$ with coefficients in $\ZZ$ has been extensively studied. In this paper we study the types of $P$ and $D(P)$ in the first homology group $H_1(\TT,\FF_2)$ with the coefficients in $\FF_2$. Our considerations connect two remarkable results concerning perfect matchings 
  of toroidal square grids, namely Temperley's bijection and the 
  Arf-invariant formula.
\end{abstract}

\section{Introduction}\label{sec:intro}
A perfect matching $P$ of a graph $G$ is a collection of edges of $G$ with
  the property that each vertex is adjacent to exactly one edge in $P$. 
This paper connects two remarkable results concerning perfect matchings 
  of toroidal square grids, namely Temperley's bijection and the 
  Arf-invariant formula.
In separate sections we describe next each of the aforementioned
  two results.

\subsection{Generalized Temperley's bijection}\label{sub.temp}
Temperley~\cite{temperley74} showed a seminal correspondence between perfect
  matchings and spanning trees of planar square grids. 
This correspondence was extended by Kenyon, Propp and Wilson~\cite{kpw00} 
  to general planar graphs. The construction also applies to graphs on other surfaces, 
  such as toroidal graphs, which we now sketch.

Let $G= (V,E)$ be a graph embedded in the torus, and let 
  $G^*= (V^*, E^*)$ be its geometric dual, whose vertices are in bijection 
  with the faces of $G$. There is a natural bijection between edges of $G^*$ and edges of $G$: let $f, f'$ be two vertices of $G^*$ and let $e$ be a common boundary edge of faces $f$ and $f'$. 
Then, $G^*$ has edge $e^*= \{f,f'\}$.  
Now, let $\G=\G(G)$ 
  be the bipartite graph whose black vertices $V^B(\G)$ are $V\cup V^*$ 
  and whose white vertices $V^W(\G)$ are in bijection with $E$ (and hence
  with $E^*$). 
Edges in $\G$ are of two types: $\{v,\{v,v'\}\}$
  where $\{v,v'\}\in E$, i.e., $\{v,v'\}\in V^W(\G)$, and  
  $\{f, \{f,f'\}\}$ where $\{f,f'\} \in E^*$
  and thus also  $\{f,f'\}\in V^W(\G)$ by the bijective correspondence of 
  the edges and the dual edges. 
Hence, each white vertex $u$ of $\G$ has degree 4: two edges incident with $u$ come from interpreting $u$ as an edge of $G$, and the remaining two edges incident with $u$ come from interpreting $u$  as a dual edge of $G$.  Graphs $\G$ obtained in this way are called 
  \emph{Temperleyan} (see~\cite{kpw00}).
This paper studies toroidal square grids 
  of size $m\times n$, which for the case where $m$ and $n$
  are even positive integers can be easily seen to be Temperleyan.

A \emph{cycle-rooted spanning forest (CRSF)} of $G$ is a spanning subgraph 
  of an orientation of $G$ where each vertex has out-degree one. 
It follows that each component of a CRSF consists of exactly one oriented
  cycle, called \emph{root cycle}, and oriented trees (arborescence) 
  attached to it and oriented towards it. 
We say that a pair $(F, F^*)$ is 
  \emph{dual} if $F$ is a CRSF that spans $G$, $F^*$ is a CRSF 
  that spans $G^*$ and each $w\in V^W(\G)$ is crossed either by an 
  edge of $F$ or by an edge of $F^*$ (but not by both).   

There is a one-to-one correspondence, called \emph{generalized Temperley's 
  bijection}, between pairs $(F,F^*)$ of dual CRSF's and perfect 
  matchings of $\G$ (see~\cite{kpw00,DG15} for the construction). 
Given a perfect matching $P$ of $\G$, the corresponding pair $(F,F^*)$ 
  is constructed as follows: we first orient each edge $e\in P$ from 
  its black vertex to its white vertex, say $w$. 
Vertex $w$ of $\G$ represents edge $e(w)$ of $G$ (or of $G^*$), and 
  edge $e$ 
  can be viewed as a 'half-edge' of $e(w)$. Let $e'$ be the other 
  'half-edge' of $e(w)$. The white vertex of $e'$ is $w$ and we orient $e'$ 
  from $w$ to its black vertex.
This defines a CRSF of $\G$ which we denote by $D(P)$
  (see Figure~\ref{fig:matchingToDigraph} for an illustration). 
We observe that
  since $\match$ is a perfect matching of $\G$, each white vertex has 
  in- and out-degree $1$ in $D(P)$, and if $u$ and $v$ are the other 
  end-vertices of the edges of $D(P)$  incident with $w$ then both $u$ 
  and $v$ are either vertices of $G$ or both are vertices of $G^*$. 
Hence, $D(P)$ uniquely decomposes into a a pair $(F,F^*)$ of dual
  CRSF's.

\begin{figure}[h]
\centering
\subfigure[Matching {$\match$} represented as rectangles.]{\label{fig:matchingToDigrapha}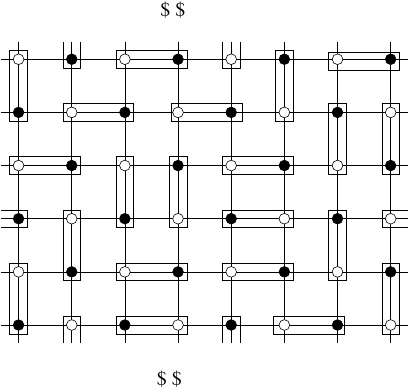}
\qquad
\subfigure[Digraph {$D(\match)$}.]{\label{fig:matchingToDigraphb}%
\begin{picture}(0,0)%
\includegraphics{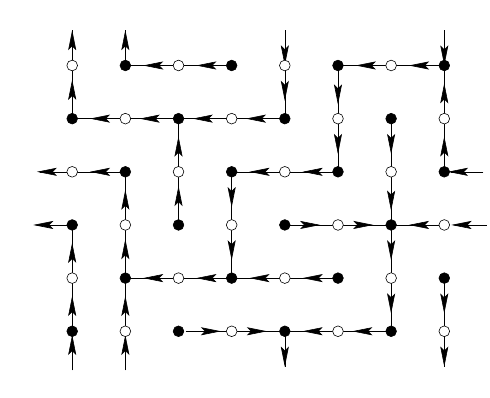}%
\end{picture}%
\setlength{\unitlength}{2486sp}%
\begingroup\makeatletter\ifx\SetFigFont\undefined%
\gdef\SetFigFont#1#2#3#4#5{%
  \reset@font\fontsize{#1}{#2pt}%
  \fontfamily{#3}\fontseries{#4}\fontshape{#5}%
  \selectfont}%
\fi\endgroup%
\begin{picture}(6192,5047)(6061,-4958)
\put(7516,-4876){\makebox(0,0)[lb]{\smash{{\SetFigFont{10}{12.0}{\familydefault}{\mddefault}{\updefault}{\color[rgb]{0,0,0}$$}%
}}}}
\put(6076,-2176){\makebox(0,0)[lb]{\smash{{\SetFigFont{10}{12.0}{\familydefault}{\mddefault}{\updefault}{\color[rgb]{0,0,0}$C_{\match}$}%
}}}}
\put(12281,-2176){\makebox(0,0)[lb]{\smash{{\SetFigFont{10}{12.0}{\familydefault}{\mddefault}{\updefault}{\color[rgb]{0,0,0}$C_{\match}$}%
}}}}
\put(7516,-106){\makebox(0,0)[lb]{\smash{{\SetFigFont{10}{12.0}{\familydefault}{\mddefault}{\updefault}{\color[rgb]{0,0,0}$$}%
}}}}
\end{picture}%
}
\caption{Perfect matching $\match$ of a $6\times 8$ grid and associated $D(\match)$.%
}\label{fig:matchingToDigraph}
\end{figure}

It is mentioned in~\cite{DG15} that, if $G$ is a toroidal
graph, then each root cycle of each of its dual CRSF's $(F, F^*)$ is
non-contractible.  We state this property for toroidal square grids formally now and postpone its proof to the next section.

\begin{lemma}\label{o.noncon}
Let $\G$ be a toroidal square grid  of size $m\times n$, 
  where both $m$ and $n$ are even. 
If $P$ is a perfect matching of $\G$ and $C$ is a root cycle of $D(P)$,
   then $C$ is non-contractible. 
\end{lemma}

It is shown in~\cite{DG15}, see also~\cite{S16}, that for toroidal graphs the generalized
  Temperley's correspondence extends to homology considerations.  
Using the construction of Thurston~\cite{Th90}, one can associate each perfect
  matching $P$ of $G$ with its \emph{height function} and with an
  element $[P]$ of the first homology group $H_1(\TT,\ZZ)$ of torus $\TT$ 
  with 
  coefficients in $\ZZ$. 
Each pair $(F,F^*)$ of dual CRSF's also naturally defines an 
  element $[F,F^*]$ of $H_1(\TT,\ZZ)$ by considering the set of the 
  root cycles.  
It is explained in~\cite{DG15} how $[P]$ and $[F,F^*]$ are related for
  a perfect matching $P$ and the corresponding dual CRSF's $(F,F^*)$.

\subsection{Arf-invariant formula}\label{sub.arf}
In statistical physics a Perfect matching is called 
  a \emph{dimer arrangement}.
Very often graph $G$ comes with an edge-weight
  function $w: E\rightarrow \QQ$. 
The statistics of the weighted perfect
  matchings is then given by the generating function
\[
P(G,x)= \sum_{\text{$\match$ perfect matching of $G$}}x^{\sum_{e\in \match}w(e)},
\] 
also known as the \emph{dimer partition function} on $G$. 
We note that if $G$ is bipartite, with bipartition 
  classes $V_{1}$ and $V_{2}$, then $P(G,x)$ equals the
  permanent of $A(G)$, where $A= A(G)$ is the $|V_1|\times |V_2|$ matrix
  given by $A_{uv}= x^{w(uv)}$.

If $g$ denotes the minimum genus of an orientable surface $\E$ in which a 
  graph $G$ embeds, then
  the partition function of a dimer model on a graph $G$ can be written
  as a linear combination of $2^{2g}$ Pfaffians of \emph{Kasteleyn matrices}.
These $2^{2g}$ Kasteleyn matrices are skew-symmetric $|V|\times |V|$
  matrices determined by $2^{2g}$ orientations of $G$, called Kasteleyn (or
  Pfaffian) orientations. 
Kasteleyn himself proved this Pfaffian
  formula in the planar case and for toroidal square grids where all
  the horizontal edge-weights and all the vertical edge-weights are the
  same~\cite{K1}, and stated the general fact~\cite{K2}. 
A complete combinatorial proof of this statement was first obtained 
  much later by Gallucio and Loebl~\cite{GL}, and independently by 
  Tesler~\cite{T}. 
Cimasoni and Reshetikhin~\cite{CR07} provided a beautiful understanding of 
  the formula in terms of the quadratic forms on $H_1(\E,\FF_2)$ and 
  since then the formula is known as the \emph{Arf-invariant formula}. 

\subsection{Main results}\label{sub.main}
The Arf-invariant formula requires to associate, 
  to each perfect
  matching $P$ of a
  graph $G$ embedded in orientable surface $\E$, an
  element $[P]_2$ of the first homology group $H_1(\E,\FF_2)$ with 
  coefficients in $\FF_2$; see e.g.~Loebl and Masbaum~\cite{LM11} for details. 
This paper studies toroidal square grids.
If $T$ is a toroidal
  square grid then $[P]_2$ can be represented by a pair $(p_1, p_2)$
  defined as follows: let $A$ be one layer of vertical edges of
  $T$ and let $B$ be one layer of horizontal edges of $T$. 
Then $p_1= |P\cap A|\mod 2$ and $p_2= |P\cap B|\mod 2$. We also say
  that $(p_1,p_2)$ is the \emph{type} of $P$. 
Let $D(P)$ be the CRSF of $P$ defined above. Again, $D(P)$ naturally
  defines an element $[D(P)]_2$ of $H_1(\E,\FF_2)$ by considering the
  root cycles: Let $C$ be a root cycle of $D(P)$. Then we represent
  $[D(P)]_2$ by a pair $(d_1,d_2)$ where $d_1= |C\cap A|\mod 2$ and
  $d_2= |C\cap B|\mod 2$.  We also say that $(d_1,d_2)$ is the \emph{type}
  of $D(P)$.
By Lemma~\ref{o.noncon} each root cycle of $D(P)$ is
  non-contractible. By construction the root cycles of $D(P)$ are disjoint and since $T$ is a toroidal graph, the type of $D(P)$ is independent of the choice of root cycle $C$. 
We next observe that the type of $D(P)$ is never $(0,0)$. For the sake of contradiction, assume $C$ is a 
 cycle in $D(P)$ of type $(0,0)$. It follows that~$C$ must be the symmetric difference of boundaries of faces 
  $F_{1},\ldots, F_{k}$ of $T$. Since $C$ is connected, it follows that $C$ is contractible which contradicts 
  Lemma~\ref{o.noncon}. We thus have: 

\begin{lemma}\label{lem:noEvenEven}
For every perfect matching $\match$ of $T$, $[D(P)]_2\neq (0,0)$.
\end{lemma}

We can now state the main result of this paper.

\begin{theorem}\label{thm.mm1}
Let $m$ and $n$ be even positive integers and let $T_{m,n}$ 
  denote the $m\times n$ toroidal square grid. 
If $P$ is a perfect matching of $T_{m,n}$ such that 
  $[P]_2\neq (0,0)$, then $[P]_2= [D(P)]_2$.
\end{theorem}

We recall that $T_{m,n}$ with both $m$ and $n$ even positive
  integers is a Temperleyan graph. 
Our proof of Theorem~\ref{thm.mm1} 
  abstracts the properties of toroidal square grids
  that are needed for the argument to go through.  
It is natural
  to ask if Theorem~\ref{thm.mm1} can be generalized to all toroidal
  Temperleyan graphs.

\section{Preliminaries}\label{sec:prelim}
Throughout this work, let $m$ and $n$ be two positive integers.
We solely consider the case where both $m$ and $n$ are even and 
  from now on, in 
  order to avoid unnecessary repetitions, we omit this condition from the 
  statement of all results.
Say $m=2m'$ and $n=2n'$.
Let $T_{m',n'}$ be the toroidal square grid with
  vertex set $\{0,\ldots,m'-1\}\times\{0,\ldots,n'-1\}$
  and edges with endpoints $(i,j)$ and
  $(i',j')$ of any of their edges differ in at most one coordinate
  and exactly by $1$,\footnote{Here and throughout this work,
  unless said otherwise,
  arithmetic involving vertex labels is always
  modulo $m$ and $n$ over the first 
  and second coordinates, respectively} i.e.,
\[
(i-i'=\pm 1 \wedge j=j')\vee (i=i'\wedge j-j'=\pm 1).
\]
Note that for the standard embedding of $T_{m',n'}$ in the 
  torus, we get that  $T^*_{m',n'}$ is isomorphic to $T_{m',n'}$.
Thus, if $\G=\G(T_{m',n'})$ is the Temperleyan graph associated to 
  $T_{m',n'}$ as described in Section~\ref{sub.temp}.
  then $\G$ is itself isomorphic to $T_{m,n}$.
Without loss of generality we can think of the $(i,j)$ nodes of $T_{m,n}$ 
  where $i$ and $j$ are of distinct parity as those in one-to-one
  correspondence with the black vertices of $\G$ (equivalently, 
  with $V(T_{m',n'})\cup V(T^*_{m',n'})$), and thus henceforth refer them
  as black vertex.
Analogously, we identify the $(i,j)$ nodes of $T_{m,n}$ of equal parity 
  with the white vertices of $\G$ (equivalently, 
  with $E(T_{m',n'})$ or $E(T^*_{m',n'})$), and also henceforth refer to them
  as white vertex.
From now on, we work directly with in $T_{m,n}$ in the understanding that 
  via the isomorphism between $\G$ and $T_{m,n}$ our discussion indeed concerns
  $\G$.

The collection of nodes $(i,j)\in V(T_{m,n})$ 
  for which $0\leq j<n$ (respectively,
  $0\leq i<m$) will be called the \emph{$i$-th row} 
  (respectively, \emph{$j$-th column}), and $i$ 
  (respectively, $j$) will be called its index.
Edges whose endvertices belong to the same row (respectively, column)
  will be called \emph{horizontal} (respectively, \emph{vertical}) edges.
We let $A$ be the complete layer of 
  horizontal edges between vertices
$(0, j)$ and $(m-1, j)$, $j= 0, \ldots, n-1$. 
Similarly, we let $B$ be the complete layer of vertical edges between vertices $(j,0)$ and $(j, n-1)$, $j= 0, \ldots, m-1$.

Paths and cycles, say $S$,  will always be considered as subgraphs
  of whatever the graph or digraph we are dealing with.
We thus rely on the standard notation $V(S)$ and $E(S)$ to 
  denote $S$'s vertices and edges, respectively.
For a simple path or cycle $S$ in whatever graph we consider, 
  we say that vertex $v\in V(S)$ is a \emph{corner} of $S$ if 
  one of the edges of $S$ incident to $v$ is a 
  horizontal edge and the other one is a vertical edge.  

We now establish a simple, yet key property, 
  of certain type of cycles of $T_{m,n}$ (equivalently, of $\G$). 
The property concerns the number of vertices encircled by such cycles.
\begin{lemma}\label{lem:contractCyclesConnected}
If $C$ is a contractible simple cycle of $T_{m,n}$ all of whose 
  corners are of the same color,
  then the disk encircled by $C$ contains a single connected component 
  of $T_{m,n}\setminus C$.
\end{lemma}
\begin{proof}
For the sake of contradiction, assume the disk encircled by
  $C$ contains two connected components of $T_{m,n}\setminus C$.
The only way in which this could happen is if 
  there is a face of $T_{m,n}$ whose boundary contains
  two edges of $C$, say $e$ and $e'$, which do not share endvertices.
Thus, $e$ and $e'$ must be either both horizontal or 
  both vertical edges.
Assume the former case happens (the latter case is handled similarly).
Consider one of the two components (paths), say $Q$, 
  obtained by removing edges $e$ and $e'$ from $C$.
Since the endvertices of $Q$ belong to two consecutive rows, then
  $Q$ must contain an odd number of column edges, but then,
  due to $T_{m,n}$'s grid like structure,  the path $Q$ can not
  have all of its corners of the same color.
\end{proof}

\begin{lemma}\label{lem:contractCyclesInterior}
If $C$ is a contractible simple cycle of $T_{m,n}$ all of whose 
  corners are of the same color,
  then the disk encircled by $C$ contains an odd number of vertices of 
  $T_{m,n}\setminus C$.
\end{lemma}
\begin{proof}
By Lemma~\ref{lem:contractCyclesConnected} we know that $C$ encircles
  one connected component of $T_{m,n}\setminus C$.
Let~$Q$~be the shortest segment of $C$ 
  (ties broken arbitrarily) whose endvertices are corners of $C$.
By minimality of $Q$, 
  all vertices of $Q$ must belong to either the same column or the same 
  row. 
Without loss of generality, assume that the former case holds and
  that $Q$'s vertices are $(i_1,j), (i_1+1,j), \ldots, (i_2,j)$.
Let $R_1$ (respectively, $R_2$) be the rectangle (cycle with four 
  corners) of $T_{m,n}$ that contains $Q$ as a segment and  
  has as opposite corners vertices $(i_1,j)$ and $(i_2,j+2)$ (respectively,
  $(i_2,j-2)$).
Because of the hypothesis concerning $C$'s corners
  and the grid structure of $T_{m,n}$, either $R_1$ or $R_2$ must encircle
  an area of the torus also encircled by $C$.
Without loss of generality, assume $R_1$ is such a rectangle and 
  denote by $C_1$ the cycle that encircles it.
Either $C=C_{1}$ or by minimality of $Q$ 
  the situation is one of the three depicted in Figure~\ref{fig:cases}.
By hypothesis, there is an odd number of vertices between any two 
  consecutive corners of $C_1$, thence $C_1$ encircles an odd number of 
  vertices.
The desired conclusion follows by induction 
  on the one or two simple cycles obtained
  by taking the symmetric difference between $C$ and $C_1$.
\end{proof}
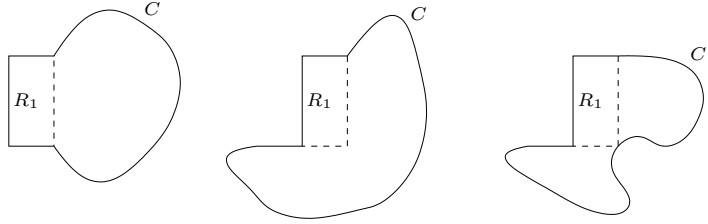
\begin{figure}[h]
\centering 
\begin{picture}(0,0)%
\includegraphics{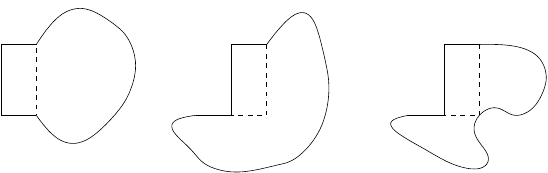}%
\end{picture}%
\setlength{\unitlength}{2486sp}%
\begingroup\makeatletter\ifx\SetFigFont\undefined%
\gdef\SetFigFont#1#2#3#4#5{%
  \reset@font\fontsize{#1}{#2pt}%
  \fontfamily{#3}\fontseries{#4}\fontshape{#5}%
  \selectfont}%
\fi\endgroup%
\begin{picture}(6948,2184)(2464,-2593)
\put(8146,-1456){\makebox(0,0)[lb]{\smash{{\SetFigFont{7}{8.4}{\rmdefault}{\mddefault}{\updefault}{\color[rgb]{0,0,0}$R_1$}%
}}}}
\put(2521,-1456){\makebox(0,0)[lb]{\smash{{\SetFigFont{7}{8.4}{\familydefault}{\mddefault}{\updefault}{\color[rgb]{0,0,0}$R_1$}%
}}}}
\put(3826,-556){\makebox(0,0)[lb]{\smash{{\SetFigFont{7}{8.4}{\familydefault}{\mddefault}{\updefault}{\color[rgb]{0,0,0}$C$}%
}}}}
\put(5446,-1456){\makebox(0,0)[lb]{\smash{{\SetFigFont{7}{8.4}{\familydefault}{\mddefault}{\updefault}{\color[rgb]{0,0,0}$R_1$}%
}}}}
\put(6481,-601){\makebox(0,0)[lb]{\smash{{\SetFigFont{7}{8.4}{\familydefault}{\mddefault}{\updefault}{\color[rgb]{0,0,0}$C$}%
}}}}
\put(9271,-1006){\makebox(0,0)[lb]{\smash{{\SetFigFont{7}{8.4}{\familydefault}{\mddefault}{\updefault}{\color[rgb]{0,0,0}$C$}%
}}}}
\end{picture}%
\caption{Possible relations between $R_1$ and $C$ when $C\neq R_{1}$.}\label{fig:cases}
\end{figure}

Consider now a perfect matching $\match$ of $T_{m,n}$
  (in contrast to our view of paths and cycles as subgraphs, we
  consider matchings as subsets of edges). We now formally define the di-graph denoted by $D(\match)$ in Section~\ref{sec:intro}. 
Since $\match$ is a perfect matching of $T_{m,n}$, we denote it by $D_{m,n}(\match)$ or simply by $D_{m,n}$.
Di-graph $D_{m,n}=D_{m,n}(\match)$ has vertex set $U$ and edge set $D$, hence $D_{m,n}=(U,D)$.
The vertex set of $D_{m,n}$ equals the vertex set of $T_{m,n}$, that is 
  $U=V(T_{m,n})$.
The edge set $D$ is such that (see Figure~\ref{fig:matchingToDigraph}):
\begin{itemize}
\item If $v$ is black, then $vw\in D$ if and only if $vw$ is an edge of $\match$.
\item If $v$ is white, then 
  $vw\in D$ provided: 
  \begin{inparaenum}[(1)]
  \item $vw$ is not an edge of $\match$, and 
  \item if $uv$ is the edge of the matching $\match$ incident to $v$, then 
    $u$, $v$, and $w$ all lie in the same row or the same column of $T_{m,n}$.
  \end{inparaenum}
\end{itemize}

We are now ready to prove  Lemma~\ref{o.noncon}.

\begin{proof}{[of Lemma~\ref{o.noncon}]}
If $C$ is a dicycle of $D_{m,n}(\match)$, 
  then $\match\setminus C$ is a perfect matching of $T_{m,n}\setminus C$. 
By definition of $D_{m,n}(\match)$ only black nodes can
  be corners, so all corners of $C$ are of the same color.
By 
  Lemma~\ref{lem:contractCyclesInterior},
  it follows that $C$ cannot be contractible since otherwise
  the set of nodes $S\subseteq V(T_{m,n})$ encircled by $C$ would be 
  of an odd cardinality set covered by $P\setminus C$.
\end{proof}


\section{Consistent Toroidal Systems}\label{sec:toroidal}
This section is dedicated to proving Theorem~\ref{thm.mm1}.
If $\match$ is a perfect matching of $T_{m,n}$ then we denote by $C_{\match}$ an arbitrary cycle of $D_{m,n}$.
Before precisely describing the abstract setting we will work in,
  we establish the following result, which essentially says that 
  we can restrict our attention to certain ``better behaved''
  perfect matchings. 
Say that a perfect matching $\match$ of $T_{m,n}$ is \emph{well behaved}, if 
  $C_{\match}$ neither contains a horizontal edge (respectively, vertical edge)
  both of whose endvertices are in row $0$ nor both in row $m-1$
  (respectively, both of whose endvertices are in column $0$ nor
  both in column $n-1$).

\begin{lemma}\label{lem:betterBehaved}
If $\match$ is a perfect matching of $T_{m,n}$,
  then there is a well behaved perfect matching $\match'$ of $T_{m+4,n+4}$
  for which the types of $\match$ and $\match'$ are the same,
  and the types of $C_{\match}$ and $C_{\match'}$ are the same.
\end{lemma}
\begin{proof}
We explicitly build $\match'$ with the desired properties (see Figure~\ref{fig:wellBehaved} for an illustration of the construction).
First, 
  for each edge in $\match\setminus (A\cup B)$, say between vertices
  $(i,j)$ and $(i',j')$ add an edge into $\match'$ between
  $(i+2,j+2)$ and $(i'+2,j'+2)$.

\begin{figure}[h]
\centering
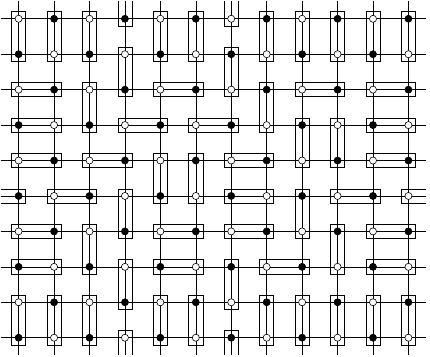
\caption{Well behaved matching $P'$ in the $10\times 12$ toroidal grid derived 
  from the matching $P$ of Figure~\ref{fig:matchingToDigrapha}.}\label{fig:wellBehaved}
\end{figure}

Now, if there is a vertical edge in $\match\cap A$ between $(0,j)$ 
  and $(m-1,j)$, then include in $\match'$ the vertical edge between $(0,j)$ 
  and $(m+3,j)$.
Note that for each column $j$ of 
  $T_{m+4,n+4}$, $2\leq j\leq n+1$, in which uncovered 
  vertices remain (there are exactly four such vertices in each column) they
  can be covered in a (uniquely determined) way by two vertical edges.
 
Similarly, if there is an edge in $\match\cap B$ between $(i,0)$ and $(i,n-1)$, 
  include in $\match'$ an horizontal edge between $(i,0)$ and $(i,n+3)$,
  and proceed as in the previous paragraph but considering rows 
  instead of columns and horizontal edges instead of vertical ones.

The vertices located in rows $0$, $1$, $m+2$ and $m+3$,
  which also belong to one of the columns $0$, $1$, $n+2$ and $n+3$,
  form $4$ groups of $4$ uncovered vertices.
Cover each such group of $4$ vertices by two vertical 
  edges and include them in $\match'$. 

It is not hard  to see that the described construction is sound, i.e., $\match'$
  is a perfect matching of $T_{m+4,n+4}$.
Moreover, $C_{\match}$ can be naturally identified with $C_{\match'}$, the latter
  of which can be easily seen to be well behaved.

Finally, let $A'$ and $B'$ be defined as $A$ and $B$ but with
  respect to $T_{m+4,n+4}$ instead of $T_{m,n}$.
It is straightforward to verify 
  that $|\match\cap A|=|\match'\cap A'|$, $|\match\cap B|=|\match'\cap B'|$, 
  $|C_{\match}\cap A|=|C_{\match'}\cap A'|$ and $|C_{\match}\cap B|=|C_{\match'}\cap B'|$,
  thus establishing the lemma's statement.
\end{proof}  

We now introduce the abstract setting in which we will henceforth work.
Let $T$ denote the surface of a torus. 
We say that $(R, N, \{A, B\}, \calP_{A,B}, C, O(C), O(R))$ is a \emph{toroidal system} 
  if the following holds (see Figure~\ref{fig:rectangle}):
\begin{itemize}
  \item \textbf{Rectangle:} $R$ is a rectangle of the torus $T$. 
  The \emph{sides} of $R$ are denoted clockwise 
  $H_{t}, V_{l}, H_{b}, V_{r}$.\footnote{Here $t,l,b$ and $r$ stand for top, left, bottom and right, respectively.} We call $H_t, H_b$ the 
  \emph{horizontal sides} and $V_l, V_r$ the \emph{vertical sides} of $R$.

  \item \textbf{Vertices:} $N$ is a set of vertices. 
  An even number of \emph{vertices} are located in each of the 
  four boundary sides of $R$.
  We denote by $A_t, A_b$ (respectively, $B_l, B_r$)
  the vertices located in $H_t, H_b$ (respectively, $V_l, V_r$). 
  It must hold that $|A_t|= |A_b|$ and $|B_l|= |B_r|$. 
Moreover, one vertex is located in each corner of $R$.
Hence, e.g., $|A_t\cap B_l|= 1$. 
These vertices are called \emph{corner vertices}.
 
  Vertices are bi-colored black and white alternately along the 
  boundary of $R$. 
We henceforth assume that the corner vertex located at the intersection of 
  $V_l$ and $H_t$ is colored white and hence completely fix
  the coloring of $A_t,B_l,A_b,B_r$. 
Traversing $H_t, H_b$ starting from $V_l$, 
  (respectively, $V_l, V_r$ starting from $H_t$) 
  one can naturally pair the vertices of 
  $A_t, A_b$ (respectively, $B_l, B_r$) in order of their appearance.
  We call \emph{opposite} vertices each such pair and note that opposite vertices have different color.

\item \textbf{Edges:} Sets $A$ and $B$ are collections of disjoint closed
  curves on $T$. Elements of $A\cup B$ are called \emph{edges}. 
  Each edge in $A$ (respectively, $B$) connects two opposite vertices of
  $A_t, A_b$ (respectively, $B_r, B_l$) and is otherwise disjoint with $R$.
Each edge of $C\cap A$ (respectively, $C\cap B$) has its black
  vertex in $A_t$ and its white vertex in $A_b$ (respectively, black vertex 
  in $B_r$ and its white vertex in $B_l$).

\item \textbf{Outer Matching:}
  $\calP_{A,B}$ is a collection of edges of $A\cup B$ referred to 
  as \emph{matching edges}. 

\item \textbf{Toroidal cycle}:
  $C$ is a non-intersecting non-contractible 
  toroidal loop which is a subset of $R\cup A\cup B$ and avoids
  every corner vertex of $R$. Moreover, each point of $C$ belonging to a side of $R$ is a vertex of $R$. 
  
\item \textbf{Cycle Matching}: $O(C)$ is a collection of disjoint closed 
  intervals of $C\cap R$ such that 
  each element of $O(C)$ is disjoint with each edge of $\calP_{A,B}$,
  has a vertex at one of its endpoints and the other endpoint 
  is a point of the interior of $R$.
Moreover, each vertex in $C$ is covered by $\calP_{A,B}\cup O(C)$.

\item \textbf{Inner Matching}: $O(R)$ is a collection of disjoint
  closed curves contained in $R$ which cover each vertex uncovered by
  $\calP_{A,B}\cup C$ and such that each curve $o$ of $O(R)$ satisfies: 
  \begin{inparaenum}[(1)] 
  \item $o$ is disjoint with $\calP_{A,B}\cup C$, and 
  \item $o$ is the segment of a side of $R$
    between two neighboring vertices, or $o$ intersects the boundary of
    $R$ in exactly one point, which is then a vertex.  
  \end{inparaenum}
\begin{figure}[h]
\centering
\subfigure[{%
  Depiction in the torus $T$ of rectangle $R$ 
  (horizontal sides $H_t$ and $H_b$ shown as
  dashed lines, and vertical sides $V_l, V_r$ shown as solid lines).
  Dotted lines correspond to edges $A\cup B$.}]{\label{fig:rectangle1}%
\begin{picture}(0,0)%
\includegraphics{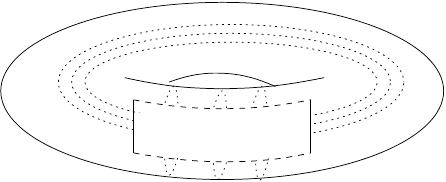}%
\end{picture}%
\setlength{\unitlength}{3108sp}%
\begingroup\makeatletter\ifx\SetFigFont\undefined%
\gdef\SetFigFont#1#2#3#4#5{%
  \reset@font\fontsize{#1}{#2pt}%
  \fontfamily{#3}\fontseries{#4}\fontshape{#5}%
  \selectfont}%
\fi\endgroup%
\begin{picture}(4516,1824)(1343,-2778)
\put(3016,-2176){\makebox(0,0)[lb]{\smash{{\SetFigFont{9}{10.8}{\familydefault}{\mddefault}{\updefault}{\color[rgb]{0,0,0}$H_t$}%
}}}}
\put(3646,-2536){\makebox(0,0)[lb]{\smash{{\SetFigFont{9}{10.8}{\familydefault}{\mddefault}{\updefault}{\color[rgb]{0,0,0}$H_b$}%
}}}}
\put(4546,-2446){\makebox(0,0)[lb]{\smash{{\SetFigFont{9}{10.8}{\familydefault}{\mddefault}{\updefault}{\color[rgb]{0,0,0}$V_r$}%
}}}}
\put(2480,-2498){\makebox(0,0)[lb]{\smash{{\SetFigFont{9}{10.8}{\familydefault}{\mddefault}{\updefault}{\color[rgb]{0,0,0}$V_l$}%
}}}}
\end{picture}%
}
\qquad 
\subfigure[%
  {Depiction of a slice of the torus.
  Bi-colored vertex sets $A_t,B_r,A_b,B_l$ located 
     in the solid straight lines representing the boundary of $R$.
  Dotted curves represent edges in $A\cup B$.
  Double solid lines correspond to edges in $\calP_{A,B}$. 
  Cycle $C$ is represented by the dashed closed curve.}]{\label{fig:rectangle2}%
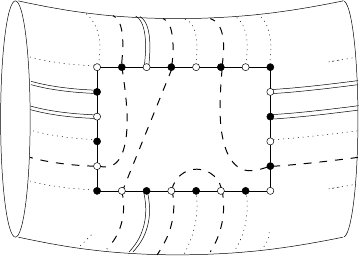}
\caption{\mbox{}}\label{fig:rectangle}
\end{figure}

\item \textbf{Segments}:
Connected components of $C\setminus (A\cup B)$ 
  (respectively, $C\setminus A$, $C\setminus B$) are called \emph{segments} 
  (respectively, $A$-segments, $B$-segments) of $C$ 
  (see Figure~\ref{fig:parts}).
The endpoints of a segment are the only one of its points
  that can belong to the boundary of $R$ (it follows that 
  segments are disjoint).
We will say that a 
  \emph{segment is bi-colored} when located at its endpoints 
  there are vertices of distinct color, otherwise we say the 
  segment is \emph{mono-chromatic}.

We require that the vertices located at the 
  endpoints of $A$-segments (respectively, $B$-segments)
  belong to $A_t$ if they are black and 
  to $A_b$ if they are white (respectively, to $B_r$ if they are black and 
  to $B_l$ if they are white).
\end{itemize}

We define $\calS(\match)=(R,N,\{A,B\},\calP_{A,B},C,O(C),O(R))$, for 
  a well behaved perfect matching $\match$, as follows.
We let $R$ equal the rectangle of the torus with sides the edges 
  in rows $0$ and $m-1$ of $T_{m,n}$, and columns $0$ and $n-1$ of 
  $T_{m,n}$.
We let $N$ be the set of vertices of $T_{m,n}$ located in the boundary 
  of $R$.
The set 
  $A$ (respectively, $B$) corresponds to the set of edges
  of $T_{m,n}$ between vertices in row $0$ and $m-1$ (respectively, $0$ and $n-1$).
The matching $\calP_{A,B}$ will be $\match\cap (A\cup B)$ and the toroidal cycle 
  $C$ is set to $C_{\match}$.
The set $O(C)$ will consist of vertical edges of $C_{\match}\setminus A$ 
  incident 
  to vertices in rows $0$ or $m-1$ that are not covered by edges of $\calP_{A,B}$
  as well as horizontal edges of $C_{\match}\setminus B$ incident to vertices
  in columns $0$ or $n-1$ also not covered by $\calP_{A,B}$.
Finally, the set of $O(R)$ will consist of all edges $\match\setminus (A\cup B)$ 
  not in $O(C)$ which are incident to vertices in rows $0$ or $m-1$,
  or to vertices in columns $0$ and $n-1$.
The following result may be easily verified.

\begin{figure}[h]
\centering
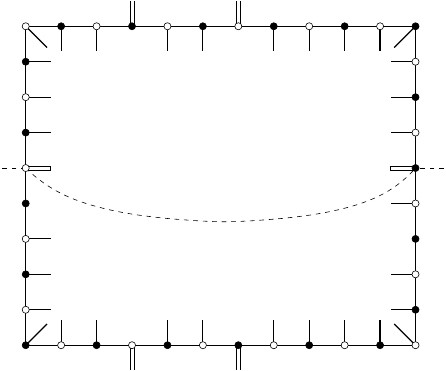
\caption{Toroidal system $\calS(P)$ associated 
  to the well behaved matching $P'$ of Figure~\ref{fig:wellBehaved}
  (boundary of rectangle $R$ drawn as solid line, 
    cycle $C$ shown as dashed line, elements of $O(R)$ shown
  as thick lines, elements of $\calP_{A,B}$ depicted as double lines outside
  $R$ and elements of $O(C)$ depicted also as double lines but inside $R$).}\label{fig:exampleToroidalSystem}
\end{figure}

\begin{proposition}\label{prop:def}
If $\match$ is a well behaved perfect matching of $T_{m,n}$,
  then $\calS(\match)$ is a toroidal system.
\end{proposition}

In fact, we will soon see that the toroidal system guaranteed by 
  the preceding proposition satisfies additional properties.
Below, we first introduce some new terminology and then
  the properties of toroidal systems that we will be concerned with.

An interval $I$ of the loop $H'_t$ (respectively, $H'_b, V'_l, V'_r$) 
  consisting of $H_t$ (respectively, $H_b, V_l, V_r$) and the edge between 
  the endvertices of $H_t$ (respectively, $H_b, V_l, V_r$) will be 
  called \emph{loop interval}. 
 
\begin{figure}[h]
\centering
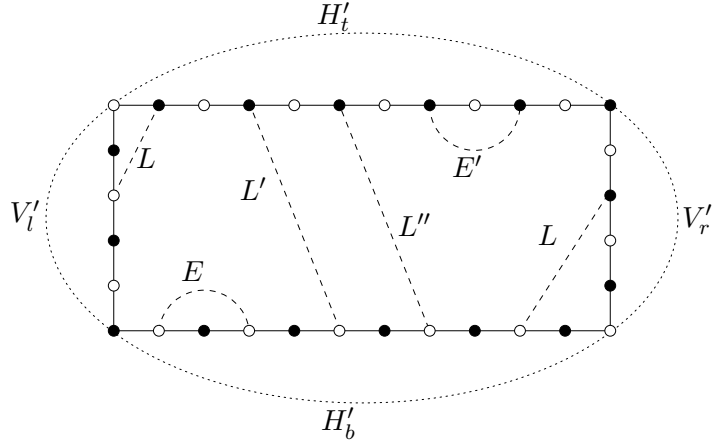
\caption{Illustration of $C\cap R$ for a toroidal cycle 
  $C$ and a rectangle $R$ 
  (the boundary of $R$ is depicted as a solid line and the cycle $C$ by dashed lines/curves): $C$ has
  six segments, five $A$-segments ($A$-lines $L$, $L'$ and $L''$ and 
  ears $E$ and $E'$) and one $B$-segment. 
  The curved dotted lines are the 
  completions of $H_t,V_r,H_b,V_l$ that form $H'_t,V'_r,H'_b,V'_l$, 
  respectively.}\label{fig:parts}
\end{figure}

For a segment
  we distinguish two situations (see again Figure~\ref{fig:parts}):
\begin{itemize}
\item \textbf{Ears:}  if the endpoints of the segment
  belong to the same side of $R$ ($H_t$, $V_r$, $H_b$, or $V_l$).

\item \textbf{Lines:} if the endpoints of 
  the segment belong to 
  opposite sides of $R$ ($H_t$ and $H_b$, or $V_l$ and $V_r$).
\end{itemize}
We analogously define $A$-ears, $B$-ears and $A$-lines, $B$-lines.
Let $D$ be a disc of $T$ bounded by exactly two $A$-lines
  and two loop intervals belonging to
  opposite loops $H'_t$ and $H'_b$.
We call $D$ \emph{$A$-regular} and its boundary loop intervals
  are called \emph{ends} of $D$. 
A toroidal cycle $C$ is called \emph{$A$-regular toroidal cycle} if all 
  its $A$-segments are $A$-lines. 
We define \emph{$B$-regular disc} and \emph{$B$-regular toroidal cycle}
  analogously. 
For instance, in Figure~\ref{fig:parts}, the cycle 
  is $B$-regular but not $A$-regular (because there are 
  two $A$-segments which are ears).
The cycle $C$ is called \emph{regular}, if it is both $A$-regular and 
  $B$-regular.

A toroidal system $(R, N, \{A, B\}, \calP_{A,B}, C, O(C), O(R))$ is called 
  \emph{consistent}, if the following properties hold:

\medskip\hangindent=\parindent%
\textbf{Segment consistency (S1):}
  Let $S$ be a segment of $C$ and let $e$ and $e'$ be 
  the edges whose endpoints are incident to 
  the endvertices of $S$. Then, $S$ is bi-colored if and only if 
  $|\{e, e'\}\cap \calP_{A,B}|$ is even.

\medskip\hangindent=\parindent%
\textbf{Line consistency (L1):}
  If $L$ is an $A$-line (respectively, $B$-line), then it has its black
  endvertex in $A_t$ and its white endvertex in $A_b$ (respectively, 
  black endvertex in $B_r$ and its white endvertex in $B_l$). 
\medskip

\textbf{Ear consistency:} Every ear $E$ is mono-chromatic and satisfies:
\begin{myquote}
\textbf{Condition (E1):} 
If $E$'s endvertices are in $A_t$ (respectively, $B_r$), 
  then their color is black, otherwise their color is white.

\item 
\textbf{Condition (E2):} 
If $D$ is the open disc of $T$ bounded by the $A$-ear (respectively, $B$-ear)
  $E$ and the loop-interval between $E$'s endvertices, 
  then $E$'s endvertices are black if and only if  
  $D$ contains an odd number of endpoints of elements of 
  $O(R)\cup O(C)\cup \calP_{A,B}$.
\end{myquote}

\textbf{Region consistency:} 
\begin{myquote}
\textbf{Condition (R1):}
If $D$ is an open $A$-regular (respectively, $B$-regular) 
  disc, then $D$ contains an
  even number of endpoints of elements of $O(R)\cup O(C)\cup \calP_{A,B}$. 

\item
\textbf{Condition (R2):}
Let $L$ be an $A$-line containing $k\geq 0$ edges of $B$. Let $R'$ be obtained by concatenation of $k+1$ copies of $R$ by its vertical boundaries. This completely determines  $N', A', B', \calP'_{A',B'}, O(R')$ and $O(C')$. 
Let $L'$ denote the cover of $L$ drawn in $R'$: we note that $L'$ has no edge of $B'$. 
If $D$ is an open disc bounded by $L'$, 
  an interval of $H'_b$ (respectively, an interval of $V'_l$),
  the vertical side $V'_r$ (respectively, horizontal side $H'_t$),
  and an interval of $H'_t$ (respectively, $V'_r$),
  then $D$ contains an even number of endpoints of
  elements of $O(R')\cup O(C')\cup \calP'_{A',B'}$.
We also require the analogous condition to hold for $B$-lines.
\end{myquote}

It is not hard to see that the following strengthening of 
  Proposition~\ref{prop:def} holds.
\begin{proposition}\label{prop:defConsistent}
If $\match$ is a well behaved perfect matching of $T_{m,n}$,
  then $\calS(\match)$ is a consistent toroidal system.
\end{proposition}

We say that $C$ has \emph{type} $\ceo$ if $C$ has an even number of edges
  of $A$ and an odd number of edges of $B$. 
Analogously, we define types $\cee, \coe, \coo$. 
As usual, a type is also associated to the set $\calP_{A,B}$ according to the 
  parity of its intersection with $A$ and $B$, respectively.

Our goal for the remaining part of this article is to prove 
  the following result (which 
  immediately implies Theorem~\ref{thm.mm1}:

\begin{theorem}\label{thm.haha}
  Let $(R, N, \{A, B\}, \calP_{A,B}, C, O(C), O(R))$ be a consistent toroidal system. If neither the
  type of $C$ nor the type of $\calP_{A,B}$ is $(0,0)$, then the type of $C$ is the
  same as the type of $\calP_{A,B}$.
\end{theorem}

Henceforth, we say that two toroidal systems
  $\calS=(R, N, \{A, B\}, \calP_{A,B}, C, O(C), O(R))$ and $\calS'=(R', N',
  \{A', B'\}, \calP'_{A',B'}, C', O(C'), O(R'))$ are \emph{equivalent} if the type
  of $C$ is the same as the type of $C'$ and the type of $\calP_{A,B}$ is the
  same as the type of $\calP'_{A',B'}$.

Theorem~\ref{thm.haha} will be a direct consequence 
  of the the next two results.
\begin{proposition}\label{p.1}
If $\calS=(R, N, \{A, B\}, \calP_{A,B}, C, O(C), O(R))$ is a consistent toroidal 
  system where neither the type of $C$ nor the type of $\calP_{A,B}$ is $(0,0)$, 
  then either Theorem~\ref{thm.haha} holds for 
  $\calS$ or there is a regular consistent
  toroidal system $\calS'$ which is equivalent to $\calS$.
\end{proposition}
\begin{proof}
We describe how to iteratively construct a new $A$-regular 
  consistent toroidal system, say 
  $\calS'= (R', N', \{A', B'\}, \calP'_{A',B'}, C', O(C'), O(R'))$, 
  equivalent to $\calS$.
The system $\calS'$ is then analogously adjusted (arguing with 
  respect to $B$ instead of $A$) to yield the 
  claimed consistent regular system.

Let $E$ be a black-endvertices-$A$-ear which is \emph{minimal}, i.e., 
  $C$ does not intersect the loop interval between the endvertices of $E$. 
Let $e_1$ and $e_2$ be the edges of $A\cap C$ which are adjacent to $E$. 
Let the black-endvertex of $e_i$ be denoted by $x_i$, and 
  the white-endvertex by $y_i$, $i=1,2$. 
Denote by $D$ the open disk with boundary $E$ and the loop interval between
  $x_1$ and $x_2$.
By the choice of~$E$, the loop intervals between $x_1, x_2$ and between 
  $y_1, y_2$ do not intersect $C$.
Because of consistency conditions (S1) and (E1), 
  there is exactly one edge of $\calP_{A,B}\cap C$ incident
  with an endvertex of $E$: without loss of generality let it be $e_1$. Let
us denote by $S_1$ the segment of $C$ prolonging $e_1$ and by $S_2$ the
segment of $C$ extending $e_2$.

Next, we establish the following claim; 
  the set of edges of $\calP_{A,B}$ that are incident with the 
  interior of the loop interval between $x_1$ and $x_2$, henceforth
  denoted $\calP_{A,B}(E)$, has an even cardinality.
Indeed, because  (E2) holds for $E$, there is an odd number of endpoints of 
  elements of $O(R)\cup O(C)\cup \calP_{A,B}$ contained in $D$.
By minimality of $E$, none of the elements 
  of $O(C)$ is contained in $D$
  and each element of $\calP_{A,B}$
  is either contained in $D$ or disjoint with it.
Hence, the interior of the loop interval between $x_1$ and $x_2$ is 
  incident with an odd number of curves of $O(R)$; 
  let us denote by $R(E)$ the set consisting of these curves.
All (the odd number of) vertices in the 
  interior of the loop interval between $x_1$ and $x_2$ 
  not covered by elements of $R(E)$ must be covered by elements of $\calP_{A,B}(E)$.
We conclude that $|\calP_{A,B}(E)|+|R(E)|$ and $|R(E)|$ are both odd, and hence
  $|\calP_{A,B}(E)|$ is even, thus proving the claim.

There are three possibilities for the types of $S_1$ and $S_2$: 
  \begin{inparaenum}[(1)]
  \item\label{it:case1}
     $S_1$ is an $A$-line and $S_2$ is a white-endvertices-$A$-ear, 
  \item\label{it:case2} 
     both $S_1$ and $S_2$ are  $A$-lines, or
  \item\label{it:case3}  
     both $S_1$ and $S_2$ are white-endvertices-$A$-ears. 
Each possible case is illustrated in Figure~\ref{fig:sit1}, \ref{fig:sit2},
  and~\ref{fig:sit3}: note that for 
  the last two cases there are two subcases 
  depending on the relative positions of $S_1$ and $S_2$.
  \end{inparaenum}
\begin{figure}[h]
\centering
\subfigure[{Case~\eqref{it:case1}.}]{\label{fig:sit1}
\begin{picture}(0,0)%
\includegraphics{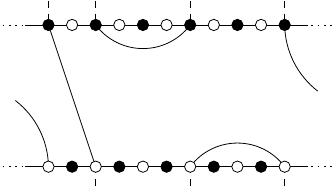}%
\end{picture}%
\setlength{\unitlength}{3315sp}%
\begingroup\makeatletter\ifx\SetFigFont\undefined%
\gdef\SetFigFont#1#2#3#4#5{%
  \reset@font\fontsize{#1}{#2pt}%
  \fontfamily{#3}\fontseries{#4}\fontshape{#5}%
  \selectfont}%
\fi\endgroup%
\begin{picture}(3174,1824)(1114,-2998)
\put(3196,-2446){\makebox(0,0)[lb]{\smash{{\SetFigFont{10}{12.0}{\familydefault}{\mddefault}{\updefault}{\color[rgb]{0,0,0}$S_2$}%
}}}}
\put(1891,-2266){\makebox(0,0)[lb]{\smash{{\SetFigFont{10}{12.0}{\familydefault}{\mddefault}{\updefault}{\color[rgb]{0,0,0}$S_1$}%
}}}}
\put(2431,-1816){\makebox(0,0)[lb]{\smash{{\SetFigFont{10}{12.0}{\familydefault}{\mddefault}{\updefault}{\color[rgb]{0,0,0}$E$}%
}}}}
\end{picture}%
}

\subfigure[{Case~\eqref{it:case2}.}]{\label{fig:sit2}
\begin{picture}(0,0)%
\includegraphics{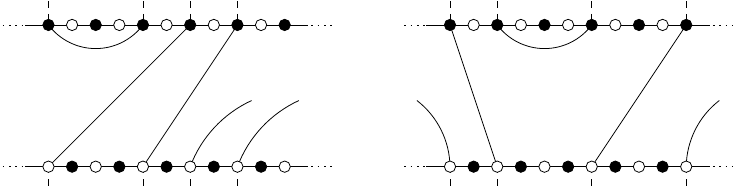}%
\end{picture}%
\setlength{\unitlength}{3315sp}%
\begingroup\makeatletter\ifx\SetFigFont\undefined%
\gdef\SetFigFont#1#2#3#4#5{%
  \reset@font\fontsize{#1}{#2pt}%
  \fontfamily{#3}\fontseries{#4}\fontshape{#5}%
  \selectfont}%
\fi\endgroup%
\begin{picture}(6999,1824)(1114,-2998)
\put(5716,-2266){\makebox(0,0)[lb]{\smash{{\SetFigFont{10}{12.0}{\familydefault}{\mddefault}{\updefault}{\color[rgb]{0,0,0}$S_1$}%
}}}}
\put(6256,-1816){\makebox(0,0)[lb]{\smash{{\SetFigFont{10}{12.0}{\familydefault}{\mddefault}{\updefault}{\color[rgb]{0,0,0}$E$}%
}}}}
\put(1846,-1816){\makebox(0,0)[lb]{\smash{{\SetFigFont{10}{12.0}{\familydefault}{\mddefault}{\updefault}{\color[rgb]{0,0,0}$E$}%
}}}}
\put(7066,-2446){\makebox(0,0)[lb]{\smash{{\SetFigFont{10}{12.0}{\familydefault}{\mddefault}{\updefault}{\color[rgb]{0,0,0}$S_2$}%
}}}}
\put(1846,-2266){\makebox(0,0)[lb]{\smash{{\SetFigFont{10}{12.0}{\familydefault}{\mddefault}{\updefault}{\color[rgb]{0,0,0}$S_1$}%
}}}}
\put(3061,-1996){\makebox(0,0)[lb]{\smash{{\SetFigFont{10}{12.0}{\familydefault}{\mddefault}{\updefault}{\color[rgb]{0,0,0}$S_2$}%
}}}}
\end{picture}%
}

\subfigure[{Case~\eqref{it:case3}.}]{\label{fig:sit3}
\begin{picture}(0,0)%
\includegraphics{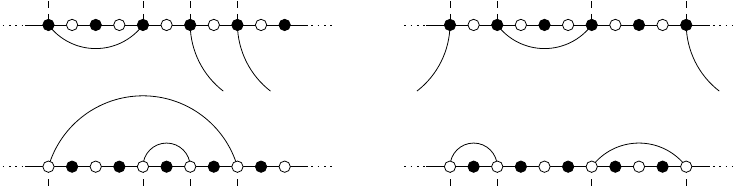}%
\end{picture}%
\setlength{\unitlength}{3315sp}%
\begingroup\makeatletter\ifx\SetFigFont\undefined%
\gdef\SetFigFont#1#2#3#4#5{%
  \reset@font\fontsize{#1}{#2pt}%
  \fontfamily{#3}\fontseries{#4}\fontshape{#5}%
  \selectfont}%
\fi\endgroup%
\begin{picture}(6999,1824)(1114,-2998)
\put(6256,-1816){\makebox(0,0)[lb]{\smash{{\SetFigFont{10}{12.0}{\familydefault}{\mddefault}{\updefault}{\color[rgb]{0,0,0}$E$}%
}}}}
\put(1846,-1816){\makebox(0,0)[lb]{\smash{{\SetFigFont{10}{12.0}{\familydefault}{\mddefault}{\updefault}{\color[rgb]{0,0,0}$E$}%
}}}}
\put(1711,-2266){\makebox(0,0)[lb]{\smash{{\SetFigFont{10}{12.0}{\familydefault}{\mddefault}{\updefault}{\color[rgb]{0,0,0}$S_1$}%
}}}}
\put(2521,-2491){\makebox(0,0)[lb]{\smash{{\SetFigFont{10}{12.0}{\familydefault}{\mddefault}{\updefault}{\color[rgb]{0,0,0}$S_2$}%
}}}}
\put(5491,-2491){\makebox(0,0)[lb]{\smash{{\SetFigFont{10}{12.0}{\familydefault}{\mddefault}{\updefault}{\color[rgb]{0,0,0}$S_1$}%
}}}}
\put(7021,-2491){\makebox(0,0)[lb]{\smash{{\SetFigFont{10}{12.0}{\familydefault}{\mddefault}{\updefault}{\color[rgb]{0,0,0}$S_2$}%
}}}}
\end{picture}%
}

\caption{\mbox{}}
\end{figure}

By consistency conditions (L1) and (E1), the vertices
  located at the endpoints of $S_1e_1Ee_2S_2$ are colored black and white,
  respectively, in case~\eqref{it:case1}; 
  are both black, in case~\eqref{it:case2}; and, are both white, in 
  case~\eqref{it:case3}.

\begin{figure}[h]
\centering
\subfigure[{Illustration of a portion of a toroidal system $\calS$ containing 
  a black-endvertices-$A$-ear labeled $E$.
  Cycle $C$ shown as dashed line, edges in $A$ shown as dotted lines, elements 
  of $O(R)$ shown as thick lines, 
  elements of $\calP_{A,B}$ depicted as double lines outside $R$, and elements 
  of $O(C)$ depicted also as double lines but inside $R$.}]{\label{fig:transf}
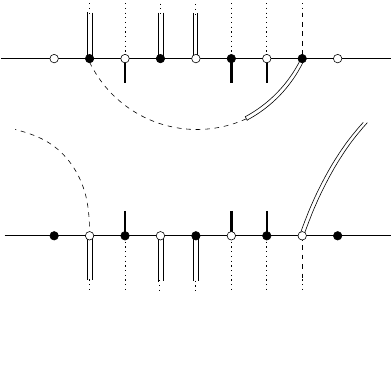}
\qquad
\subfigure[{Illustration of the portion of the toroidal system 
  $\calS'$ obtained by applying (to the system $\calS$ illustrated in 
  Figure~\ref{fig:transf}) the procedure described in the proof of 
  Proposition~\ref{p.1}. The components $C'$, $\calP'_{A',B'}$, and the elements of 
  $O(R')$ and  $O(C')$ are depicted following the same conventions as 
  in Figure~\ref{fig:transf}.}]{\label{fig:transfb}
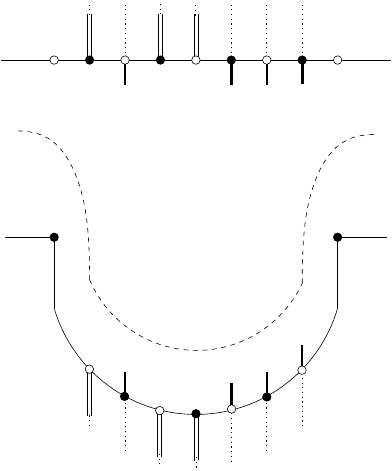}
\caption{\mbox{}}
\end{figure}

Let $z_i$ be the neighbor of $y_i$ along $H_b$ (see Figure~\ref{fig:transf}). 
We construct $R'$ by attaching a new disc $A_b(E)$ to the loop 
  interval between $z_1$ and $z_2$; in the interior of the outer boundary of 
  $A_b(E)$ we place the same number of vertices as there 
  are in the loop interval of $R$ 
  between $z_1$ and $z_2$ (see again Figure~\ref{fig:transf}). 
We then copy $e_1Ee_2$ to the just added disc and 
  obtain $C'$ from $C$ as shown in Figure~\ref{fig:transfb}. 
Next, we remove vertices that end up in the interior of $R'$.
This completely determines the new vertex set $N'$. 
Every $A$-edge $e$ incident to vertices $x'_i$ and $y'_i$, 
  $x'_i$ between $x_1$ and $x_2$ (including them) and
  $y'_i$ between $y_1$ and $y_2$ (also including them),
  is replaced by an edge $e'$ with endvertices 
  $x'_i$ and the copy of $y'_i$ that was placed in the outer boundary
  of $A_b(E)$.
All other edges are preserved.
This completely determines the new set of edges $A'$ and $B'$.
Note that there is a one-to-one correspondence between old $A/B$-edges and 
  new $A'/B'$-edges.
The set of new edges associated to edges in $\calP_{A,B}$ determine $\calP'_{A',B'}$. 
The curves of $O(C')$ are those elements of $O(C)$ contained in $C'$
  which are incident to the boundary
  of $R'$ (see Figure~\ref{fig:transfb}).
Finally, we add to $O(R')$ all elements in 
  $O(R)$ that are incident to the boundary of $R'$ 
  as well as curves that match all the
  not-yet-covered new vertices (see Figures~\ref{fig:transf} 
  and~\ref{fig:transfb}).

We need to show that the new system $\calS'$ is consistent.
First, note that $S=S_1e_1Ee_2S_2$ is the only new segment created. 
In case~\eqref{it:case1}, 
  $S$ is a line whose endvertex in $A_t$ is the corresponding
  vertex of $S_1$ (hence, by consistency condition (L1) of $S_1$, 
  it is colored black) and whose other endvertex
  in $A_b$ is the corresponding one of $S_2$ (hence, by (E1) of
  $S_2$, it is colored white).
It follows that the line $S$ satisfies condition (L1).
In case~\eqref{it:case2}, $S$ is a black-endvertices-$A$-ear with
  endvertices in $A_t$; and, in case~\eqref{it:case3}, 
  $S$ is a white-endvertices-$A$-ear with endvertices 
  in $A_b$. 
Thus, for these latter two cases, condition (E1) is satisfied.

We show next that $S$ satisfies consistency condition (S1).
Let $e'_1$ be the edge incident to the endvertex of $S_1$ not covered by $e_1$.
Similarly, let $e'_2$ be the edge incident to the endvertex of $S_2$ not
  covered by $e_2$.
Because consistent lines are bi-colored and consistent ears are monochromatic,
  we have that $S$ is bicolored if and only if case~\eqref{it:case1} holds. 
Thus, to establish that (S1) holds in $\calS'$ we need 
  to show that $|\{e'_1,e'_2\}\cap \calP'_{A',B'}|$ is even if and only if 
  case~\eqref{it:case1} holds.
Since by condition (E1), the term $|\{e_1,e_2\}\cap \calP_{A,B}|$ is odd, 
  we get that the parity of $|\{e'_1,e'_2\}\cap \calP'_{A',B'}|$ differs
  from the parity of $|\{e'_1,e_1\}\cap \calP_{A,B}|+|\{e_2,e'_2\}\cap \calP_{A,B}|$,
  which by consistency conditions (L1) and (E1) of $\calS$ is odd
  if and only if case~\eqref{it:case1} occurs.
This finishes the proof that $\calS'$ satisfies (S1).

It is important to note that regarding $B$-ears and $B$-regular discs, 
  we did not change the system. 
Hence, $B$-ears and open $B$-regular discs stay consistent. 

Next, we show that $A$-ears and regions of 
  $\calS'$ are also consistent. 
First, let $\widetilde{E}$ be an $A$-ear of $\calS'$ which does not involve 
  the new segment $S$. 
Clearly, $\widetilde{E}$ is an $A$-ear in $\calS'$ as well.
Observe that unless the disc $\widetilde{D}$ of the consistency
  condition (E2) of $\widetilde{E}$ intersects $A_b(E)$, condition
  (E2) is clearly satisfied.
Thus, assume $\widetilde{D}$ intersects $A_b(E)$.
Note that by minimality
  of $E$, $|R(E)|+1$ endpoints are removed from the interior of 
  $\widetilde{D}\setminus A_{b}(E)$ and exactly the same number is 
  added to the interior of $A_{b}(E)$.
More precisely, $|R(E)|$ of the removed points are endpoints of 
  elements of $O(R)$ incident to vertices between $y_1$ and $y_2$ and 
  the remaining removed point is the endpoint of the element of 
  $O(C)$ covering $y_2$.
Similarly, $|R(E)|+1$ of the added points 
  belong to elements of $O(R')$ covering the vertices of 
  the outer boundary of $A_b(E)$.
The same reasoning establishes conditions (R1) and (R2) in 
  $\calS'$ when the region's boundary does not contain the new segment $S$.

We next argue that if $S$ forms an $A$-ear, then (E2) holds for $S$ in 
  $\calS'$.
In case~\eqref{it:case1}, no new ear is created. 
In case~\eqref{it:case3}, 
  the new segment $S$ is a white-endvertices-$A$-ear and the 
  interior of its open disc $D(S)$ consists of the interior of the open 
  disc $D(S_1)$ of $S_1$, the interior of the open disc $D(S_2)$ of $S_2$, 
  and the points on elements of $O(R')$ covering the vertices of the 
  outer boundary of $A_b(E)$; we argued above that their number is even. 
Hence, by consistency of $S_1$ and $S_2$, segment $S$ satisfies (E2).
In case~\eqref{it:case2}, the new segment $S$ is a black-endvertices-$A$-ear. 
Because (L1) and (E2) hold for~$\calS$, 
  the difference between the number of endpoints of elements of 
  $O(R)\cup O(C)\cup \calP_{A,B}$ in its interior and in the open $A$-regular 
  disc of $\calS$ defined by $S_1$ and $S_2$ is 
  exactly $|R(E)|$, which is odd.  
Hence, the new ear $S$ also satisfies (E2).

It remains to check (R1) and (R2) for the regions which contain the 
  new segment $S$ as a line in their boundary. 
This situation only arises in case~\eqref{it:case1} 
  when $S_1$ is a line and $S_2$ is a white-endvertices-$A$-ear in $\calS$. 
Next, we determine the difference of the number of the inner points 
  on each side of $S$ with respect to $S_1$. 
On one side we add the even number of endpoints of elements 
  of $O(R')$ covering the vertices of the outer boundary of $A_b(E)$. 
To the same side we add the points in the open disc $D(S_2)$ of $S_2$, 
  whose number is even by (E2).  
On the other side we delete $|R(E)|$ (odd number) endpoints covered 
  by elements of $O(R)$ incident with the interior of the loop 
  interval between $y_1$ and $y_2$. 
We further delete an even number of endpoints in the interior of $D(S_2)$, 
  and finally we delete exactly one endpoint in the interior of $S_2$
  (the endpoint in the interior of $R$ of the element of $O(C)$ incident
  to $y_2$).
Summarizing, the aforementioned difference on each side of $S$ with respect to 
  $S_1$ is even and thus both conditions  (R1) are (R2) are 
  established.
\end{proof}

Summing up, if there is a counterexample to Theorem~\ref{thm.haha},
  then there is a regular 
  counterexample~$\calS=(R, N, \{A, B\}, \calP_{A,B}, C, O(C), O(R))$.
Assume $\calS$ is such that  $|A\cap C|= k$ and $|B\cap C|= l$.  
Let $\calS'=(R', N', \{A', B'\}, \calP'_{A',B'}, C', O(C'), O(R'))$ 
  be the system obtained by first concatenating $l$ copies of 
  $\calS$ by the vertical boundaries of $R$, and then concatenating 
  $k$ copies of the resulting system by the horizontal boundaries of $R$. 
We claim that $\calS'$ 
  is a regular consistent toroidal system.
Indeed, the regularity is clear since $C'$ has exactly one edge of $A'$ and
  exactly one edge of $B'$. 
Thus, conditions (E1) and (E2) vacuously hold for $\calS'$.
Moreover, conditions (S1) and (L1) are easily seen to be inhereted 
  by~$\calS'$ from $\calS$.
Hence, the region consistency conditions
  (R1) and (R2) need to be checked and easily seen to follow from the
  consistency of $\calS$. 

If both $k$ and $l$ are odd, then the type of $\calP'_{A',B'}$ is the same as 
  the type of $\calP_{A,B}$. 
Let exactly one of $k, l$, say $k$, be even. 
Furthermore, let $k> 0$. 
Since~$\calS$ is a counterexample to Theorem~\ref{thm.haha}, 
  it must be that $\calP_{A,B}$ has an odd number of edges of $A$ and
  thence $\calP'_{A',B'}$ has an odd number of edges of $A'$.
Clearly, $\calP'_{A',B'}$ will have an even number of edges of $B'$.
Since $C'$ has exactly one edge in each of $A'$ and $B'$, 
  it follows that $\calS'$ would also be a counterexample to 
  Theorem~\ref{thm.haha}.
Summarizing, if there is a counterexample to Theorem~\ref{thm.haha},
  then there is a counterexample $\calS$ with both
  $A\cap C$ and $B\cap C$ having at most one edge.

\begin{proposition}\label{p.2}
If $\calS=(R, N, \{A, B\}, \calP_{A,B}, C, O(C), O(R))$ is a regular consistent toroidal
  system with both $A\cap C$ and $B\cap C$ consisting of at most one edge, 
  then $\calS$ satisfies Theorem~\ref{thm.haha}.
\end{proposition}
\begin{proof}
First, we consider the case where $C\cap B=\emptyset$.
Because $C$ is not of type $\cee$, 
  it must be that $|C\cap A|= 1$.
Let $L$ be the only $A$-line.
Let $D$ be the open disc bounded by $L$, an interval of $H_b$,
  the vertical side of $V_r$, and an interval of $H_t$. 
By condition (R2), we have that 
  $D$ contains an even number of 
  endpoints of elements of $O(R)\cup O(C)\cup \calP_{A,B}$.
However, only endpoints of $O(R)$ belong to $D$, and therefore $O(R)$
  covers an even number of vertices of the boundary of $D$.
The remaining (even number) of vertices must be covered by $\calP_{A,B}$. 
Since each edge of $\calP_{A,B}\cap A$ covers an even number 
  of vertices of the boundary of $D$, it follows that 
  $|\calP_{A,B}\cap B|$ is even.
Since $\calP_{A,B}$ can not be of type $\mee$, we get that $\calP_{A,B}$ is of type $\moe$,
  thus completing the analysis of the case being considered.
The case that $C\cap A=\emptyset$ is analogous.

To conclude, let $C$ have exactly one edge from both $A$ and $B$. 
We note that $C\cap \calP_{A,B}$ may consist of zero, one or two edges.
We consider in detail the case where $|C\cap \calP_{A,B}| =2$.
The remaining cases can be dealt with similarly.
Let $L$ be the only $A$-line of $C$.
As in condition (R2) for $L$, 
  we concatenate two copies of $R$ by its vertical side.
Let $A'$-line $L'$ be the cover of $L$ in the concatenated rectangle $R'$
  and let $D$ be one of the two open discs obtained
  by subtracting $L'$ from the interior of $R'$.

By condition (R2) the disc $D$ has an even number of endpoints 
  of elements of $O(R')$. 
We also note that $D$ has 
  an even number of vertices in its boundary. 
It follows that $\calP'_{A',B'}$ (see condition (R2)) covers an even 
  number of vertices of $D$'s boundary.
Moreover, it is not hard to see that edges of 
  $\calP'_{A',B'}$ that are incident to exactly one vertex of the boundary of $D$
  can be placed in one-to-one correspondence with the edges 
  of $\calP_{A,B}$ that are incident to at least one 
  vertex in the boundary of $D$.
Hence, the cardinality of $\calP_{A,B}$ must also be even. 
Since $\calP_{A,B}$ is not of type $\mee$, it must be of type $\moo$.
\end{proof}


%

\bibliographystyle{alpha}
\bibliography{biblio}

\end{document}